\documentclass[11pt, onecolumn, a4paper]{article}

\usepackage{amssymb, amsmath}
\usepackage{amsthm}

\newcommand{\nop}[1]{}

\newtheorem{theorem}{Theorem}
\newtheorem{corollary}{Corollary}
\newtheorem{lemma}{Lemma}

\begin{document}

\title{Performance Bounds for Multiclass FIFO in Communication Networks: A Deterministic Case}         
\author{Yuming Jiang}        
\maketitle

\begin{abstract}
Multiclass FIFO is used in communication networks such as in input-queueing routers/switches and in wireless networks. For the concern of providing service guarantees in such networks, it is crucial to have analytical results, e.g. bounds, on the performance of multi-class FIFO. Surprisingly, there are few such results in the literature. This paper is devoted to filling the gap. Specifically, a single hop deterministic case is studied, for which, delay and backlog bounds are derived, in addition to guaranteed rate and service curve characterizations that may be exploited to extend the analysis to network cases. 
\end{abstract}

\section{Introduction}\label{sec-1}

Multiclass FIFO refers to the scheduling discipline where packets (or customers in the queueing theory jargon) are served in the first-in-first-out (FIFO) manner and the service rates to packets of different classes may differ. 

Multiclass FIFO is perhaps the most intuitive scheduling discipline for many scenarios in communication networks. One example is input-queueing on a router/switch, where packets are FIFO-queued at each input port before being forwarded to the corresponding output ports that may have different capacities. Another example is downlink-sharing in wireless networks, where a wireless base station, shared by multiple users, sends packets to these users in the FIFO manner while the characteristics of the wireless channel seen by packets of different users may be different. 

For the concern of providing service guarantees in such networks, it is crucial to have analytical results, e.g. bounds, on the performance of multi-class FIFO. Surprisingly, there are few such results in the literature. In addition, as to be shown in the paper, the bounds that can be directly deduced from the literature results require conditions that may be too restrictive, limiting the applicability of these bounds. This paper is devoted to filling this gap. The focus is on a deterministic case. Specifically, both delay and backlog bounds are derived. In addition, the guaranteed rate and service curve characterizations for each class are also obtained, which may be further exploited to extend the analysis to network cases. 

The rest is organized as follows. In the next section, the system model and performance metrics are introduced. In Sec.~\ref{sec-3}, the direct bounds are presented and their limitations discussed. In Sec.~\ref{sec-4}, improved bounds with much relaxed conditions are derived and discussed. Concluding remarks are given in Sec.~\ref{sec-5}.

\section{The System Model}\label{sec-2} 

\subsection{The System Model}

We consider a work-conserving multiclass system serving packets in a communication network. It is a discrete-time system with time indexed by $t=0, 1, 2, \dots$, where the length of the time unit is not specified, but is supposed to be small enough to even count one bit service time. There are $N$ classes of packets. Packets are served in the FIFO manner. If multiple packets arrive at the same time, the tie is broken arbitrarily. The service rate (in bps) of each class, denoted by $C_n$, is assumed to be constant. Let $C_{min} = \min_{n}\{C_n\}$ and $C_{max} = \max_{n}\{C_n\}$.

By convention, a packet is said to have arrived to (respectively served by) the system when and only when its last bit has arrived to (respectively left) the system. When a packet arrives seeing the system busy, the packet will be queued and the buffer size for the queue is assumed to be large enough ensuring no packet loss. The queue is initially empty. 

For each class $n$, let $f_n$ denote its traffic flow and $p^{f_n,i}$ the $i$th packet $(i = 1, 2, \dots)$ of the flow. For each packet $p^{f_n,i}$, we denote by $a^{f_n,i}$ its arrival time to the system, $d^{f_n,i}$ its departure time from the system, and $l^{f_n,i}$ its length (in bits). For the {\em aggregate flow} of $f_n$, $(n=1, \dots, N)$, denoted by $g$, we use $p^{g,j}$, $a^{g,j}$, $d^{g,j}$ and $l^{g,j}$ to respectively denote the $j$th packet of the aggregate flow, its arrival time, its departure time and its length. Let $L_n$ be the maximum packet length of flow $f_n$ and $L = \max_{n}\{L_n\}$. 

In addition, we use $A_n(s,t)$ and $A_n^{*}(s,t)$ to respectively denote the amount of traffic (in bits) that has arrived and departed from flow $f_n$ within time period $[s, t]$, and let $A_n(t) = A_n(0,t)$ and $A_n^{*}(t) = A_n^{*}(0,t)$. For the aggregate flow, $A(t)$ and $A^{*}(t)$ are similarly defined.

The traffic of each class is assumed to be constrained by a leaky-bucket {\em arrival curve} \cite{NetCal} $\alpha_n(t) = r_n t + \sigma_n$, or in other words: for all $t \ge s \ge 0$, 
$$
A_n(s, t) \le \alpha_n(t-s) = r_n (t-s) + \sigma_n.
$$

\subsection{Performance Metrics of Interest}
In this paper, we focus on finding bounds on packet delay, backlog, guaranteed rate and service curve, which are defined as follows.

(i) The delay of packet $p^{f_n,i}$, denoted by $D^{f_n,i}$, is naturally 
\begin{eqnarray}
D^{f_n,i} &=& d^{f_n,i} - a^{f_n,i}.
\end{eqnarray}

\nop{
In addition, we define the virtual delay at time $t$ as 
\begin{eqnarray}
D(t) &=&  \inf \{\tau: A^{*}(t+\tau) \ge A(t) \} .
\end{eqnarray}
Due to FIFO, a bound on $D^{f_n,i}$ can also be found from the following, where the equation holds if there is no concurrent arrival at $a^{f_n,i}$ 
\begin{equation}
D^{f_n,i} \le D(a^{f_n,i}). 
\end{equation}
}

(ii) The backlog at time $t$, denoted by $B(t)$, is defined as:
$$
B(t) = A(t) - A^{*}(t).
$$

(iii) A system is said to be a Guaranteed Rate (GR) server with rate $R_n$ and error term $E_n$ to traffic class $n$ or flow $f_n$, if it guarantees that for any packet $p^{f_n,i}$ of the flow, its departure time satisfies \cite{GR}:
\begin{equation}
d^{f_n,i} \le GRC^{f_n,i}(R_n) + E_n \label{def-gr}
\end{equation}
where the guaranteed rate clock (GRC) function is iteratively defined as: 
\begin{equation} \label{eq-grc}
GRC^{f_n,i}(R_n) = max\{a^{f_n,i}, GRC^{f_n, i-1}\} + \frac{l^{f_n,i}}{R_n},
\end{equation}
with $GRC^{f_n,0} = 0$.

(iv) A system is said to provide a {\em service curve} $\beta_n(t)$ to traffic class $n$ or flow $f_n$ if there holds \cite{NetCal}: for all $t \ge 0$, 
\begin{equation}
A_n^{*} \ge A_n \otimes \beta_n (t) \label{sc}
\end{equation}
where \nop{the min-plus convolution, denoted by $\otimes$, of functions $A(\cdot)$ and $\beta(\cdot)$, is defined as:}
$A_n \otimes \beta_n (t) \equiv \inf_{0 \le s \le t}\{A_n(s) + \beta_n(t-s)\}. $

\section{Direct Results and the Limitations} \label{sec-3}
In this section, we introduce performance bounds that can be derived by making easy application of the literature results. 

\begin{lemma}\label{lm-gra}
The multiclass FIFO system is a Guaranteed Rate server with rate $R=C_{min}$ and error term $E=0$ to the aggregate input.
\end{lemma}

\begin{proof}
For the multiclass FIFO system, there holds for every packet $p^{g,j}$ of the aggregate flow, $j=1, 2 \dots$,
\begin{eqnarray}
d^{g,j} &=& \max \{a^{g,j}, d^{g,j-1} \} + \frac{l^{g,j}}{C_{*}} \label{eq-vtf}
\end{eqnarray}
where $C_{*}$ is the service rate of the class of packet $p^{g,j}$. 

Since $C_{*} \ge C_{min}$, it can be easily verified:
$$
d^{g,j} \le GRC^{g,j}(C_{min})
$$
with which, the lemma follows from the definition of GR.
\end{proof}

With Lemma \ref{lm-gra}, the following result follows from the relationship between GR and service curve \cite{Jiang03} \cite{NetCal}.

\begin{lemma}\label{lm-sca}
The multiclass FIFO system provides to the aggregate input a service curve $\beta(t) = (C_{min} \cdot t - L)^{+}$.
\end{lemma}

With Lemmas \ref{lm-gra} and \ref{lm-sca}, the following corollaries follow immediately from the related literature results. Specifically, Cor.~\ref{co-db}, Cor.~\ref{co-bb} and Cor.~\ref{co-losc} follow respectively from the delay bound, backlog bound and leftover service results of the deterministic network calculus \cite{Chang00, NetCal}. In addition, Cor.~\ref{co-gr} further follows from the relationship between GR and service curve \cite{Jiang03, NetCal}.

\begin{corollary}\label{co-db}
(i) {\bf Delay Bound:} If $\sum_{n} r_n \le C_{min}$, then the delay of any packet $p^{g,j}$ is bounded by: 
$
D^{g,j} \le \frac{\sum_{n} \sigma_n}{C_{min}}.
$
\end{corollary}

\begin{corollary}\label{co-bb}
(ii) {\bf Backlog Bound:} If $\sum_{n} r_n \le C_{min}$, then the backlog at any time $t$ is bounded by:
$
B(t) \le \sum_{n} \sigma_n + (\sum_{n} r_n) \frac{L}{C_{min}} .
$
\end{corollary}


\begin{corollary}\label{co-gr}
(iii) {\bf Guaranteed Rate to Each Traffic Class:} For each traffic class $n$, if $\sum_{m \neq n} r_m < C_{min}$, then the system is a guaranteed rate server to it with rate $R_n=C_{min}-\sum_{m \neq n} r_m$ and error term
$
E_n = \frac{\sum_{m \neq n} \sigma_m + L}{C_{min}-\sum_{m \neq n} r_m}.
$
\end{corollary}

\begin{corollary}\label{co-losc}
(iv) {\bf Service Curve to Each Traffic Class:} For each traffic class $n$, if $\sum_{m \neq n} r_m < C_{min}$, then the system provides to it a service curve 
$
\beta_n(t) = [(C_{min}-\sum_{m \neq n} r_m) \cdot t - L - \sum_{m \neq n} \sigma_m]^{+}.
$
\end{corollary}

\subsection{The Limitations} 

The alert reader may have noticed that the conditions for the performance bounds presented above may be too restrictive. Specifically, in obtaining the delay and backlog bounds, it has been required that $\sum_{n} r_n \le C_{min}$. This requirement may greatly limit the applicability of these bounds particularly when $C_n$, $n=1, \dots, N$, differ much. For a simple example where there are two classes, i.e. $N=2$, suppose $C_1=1$ Mbps and $C_2=100$ Mbps. This requirement basically limits the total traffic rate of both classes to below 1 Mbps in order to apply the delay and backlog bounds shown above.  

A similar limitation is found for the guaranteed rate and service curve to each traffic class $n$, where, it is required that $\sum_{m \neq n} r_m < C_{min}$, i.e., the total traffic rate of all other classes must be smaller than the minimum capacity $C_{min}$. Take the two-class example again. Cor.~\ref{co-gr} and Cor.~\ref{co-losc} essentially suggest that the system only guarantees the 2nd class a surprisingly low long-term service rate $1-r_1$ Mbs, even though $C_2 = 100$ Mbs!

Due to these limitations, the directly obtained bounds may have limited use or are even in contrary to intuitions. In the next section, we aim to relax the limitations and derive better bounds.

\section{Main Results} \label{sec-4}

We start with the main results and then present their proofs, followed by a discussion comparing them with the direct results.

\subsection{Summary of Main Results}

\begin{theorem}\label{th-db}
(i) {\bf Delay Bound:} For the multiclass FIFO system, if  $\rho \equiv \sum_{n}\frac{r_n}{C_n} \le 1$, the delay of any packet $p^{g,j}$ is bounded by:
$$
D^{g,j} \le \sum_{n}\frac{\sigma_n}{C_n}
$$
and the delay bound is tight in the sense that it may be reached. 
\end{theorem}


\begin{theorem}\label{th-bb}
(ii) {\bf Backlog Bound:} For the multiclass FIFO system, if  $\rho \equiv \sum_{n}\frac{r_n}{C_n} \le 1$, the backlog at any time $t$ is bounded by:
\begin{eqnarray} 
B(t) &\le& \min \left \{(\sum_n r_n) \sum_n \frac{\sigma_n}{C_n} + \sum_n \sigma_n, \right. \nonumber \\
&& \left. C_{max} \cdot \sum_n \frac{\sigma_n}{C_n} + \rho C_{max} \cdot \max_{n}\frac{L_n}{C_n} \right \}. \nonumber 
\end{eqnarray} 
\end{theorem}

\begin{theorem}\label{th-gr}
(iii) {\bf Guaranteed Rate to Each Traffic Class:} If  $\bar{\rho} \equiv \sum_{m \neq n}\frac{r_m}{C_m} < 1$, the system is a guaranteed rate server to class $n$ with rate $R_n=(1-\bar{\rho}) C_n$ and error term 
$$
E_n = \sum_{m \neq n}\frac{\sigma_m}{(1-\bar{\rho})C_m}.
$$
\end{theorem}

\begin{theorem}\label{th-losc}
(iv) {\bf Service Curve to Each Traffic Class:} If  $\bar{\rho} \equiv \sum_{m \neq n}\frac{r_m}{C_m} < 1$, the system provides to class $n$ a service curve 
$$
\beta_n(t) =  [(1-\bar{\rho}) C_n \cdot t - L_n - C_n \cdot \sum_{m \neq n}\frac{\sigma_m}{C_m} ]^{+} . 
$$ 
\end{theorem}

\subsection{Proofs}

\subsubsection{Proof of Theorem \ref{th-db}}
For any packet $p^{g, j}$, there exists time $t_0$ that starts the busy period of which the packet is in. Note that such a busy period always exists, since in the extreme case, the period is only the service time period of $p^{g,j}$ and in this case, $t^0=a^{g,j}$. 

Since the system is work-conserving and is busy with serving between $t^0$ and $d^{g,j}$, there holds:
$
d^{g,j} = t^0 + \sum_{m=1}^{N}T_m(t^0, d^{g,j}),
$
where $T_m(t^0, d^{g,j})$ denotes the service time of packets from class $m$, served in $[t^0, d^{g,j}]$. 
Because of FIFO and that the system is empty at $t^0_{-}$, $T_m(t^0, d^{g,j})$ is hence limited by the amount of traffic that arrives in $[t^0, a^{f_n,i}]$: $T_m(t^0, d^{g,j}) \le \frac{A_m(t^0, a^{g,j})}{C_m}$. 

We then have, 
\begin{equation}\label{eq-2}
d^{g,j} \le t^0 + \sum_{m=1}^{N}\frac{A_m(t^0, a^{g,j})}{C_m}
\end{equation}

Under the condition that $\sum_{m=1}^{N}\frac{r_m}{C_m} \le 1$, we obtain:
\begin{eqnarray}\label{eq-2}
D^{g,j} &=& d^{g,j} - a^{g,j} \le \sum_{m=1}^{N}\frac{A_m(t^0, a^{g,j})}{C_m} + t_0-a^{g,j} \nonumber \\
&\le& \sum_{m=1}^{N}\frac{r_m(a^{g,j}-t^0) + \sigma_m}{C_m} - (a^{g,j}- t_0) \\
&\le&  \sum_{m=1}^{N} \frac{ \sigma^{f_m}}{C_m}.
\end{eqnarray}

Note that, for the system, consider that immidiately after time $0$, every traffic class generates a burst with size $\sigma_m$. In this case, the packet in the bursts, which receives service last, will experience delay $\sum_{m=1}^{N} \frac{ \sigma^{f_m}}{C_m}$ that equals the delay bound. So, the bound is tight. 
\qed

\subsubsection{Proof of Theorem \ref{th-bb}}

The backlog bound has two parts. 

The first part follows from Th. \ref{th-db} that the system provides to the aggregate a service curve $\beta(t) = \delta_{\sum_n \frac{\sigma_n}{C_n}}(t)$ \cite{NetCal}, where $\delta_D(t) = 0$ if $0 \le t \le D$ and $\delta_D(t) = \infty$ when $t > D$. With this and that 
the aggregate input is constrained by an arrival curve as:
$
A(t) \le (\sum_n r_n) \cdot t + \sum_n \sigma_n \equiv \alpha(t).
$
Then, the first part follows from the delay bound result of the deterministic network calculus \cite{NetCal}. 

To obtain the second part, we propose the following mapping method. Recall that, the system behavior is essentially determined by (\ref{eq-vtf}), which can be re-written as 
\begin{eqnarray}
d^{g,j} &=& \max \{a^{g,j}, d^{g,j-1} \} + \hat{l}^{g,j} \label{eq-vtf2} 
\end{eqnarray}
where $\hat{l}^{g,j} \equiv \frac{l^{g,j}}{C_{*}}$ and $C_{*}$ is the service rate of the class of $p^{g,j}$. 

Equation (\ref{eq-vtf2}) suggests that the multiclass FIFO system is equivalent to a reference FIFO system that has constant unit service rate for every packet whose length is equal to that in the original multiclass FIFO system normalized by the service rate of its corresponding class. Under this normalization, it can be verified that: $\hat{A}_n(s,t) = \frac{A_n(s,t)}{C_n}$, and $\hat{A}^{*}_n(s,t) = \frac{A^{*}_n(s,t)}{C_n}$ for all $n, s, t$, and $\hat{B}(t) = \sum_{n}\frac{B_n(t)}{C_n}$, where $\hat{A}_n(s,t)$ and $\hat{A}^{*}_n(s,t)$ respectively denote the (normalized) arrival and departure processes in the reference FIFO system and $\hat{B}(t)$ the backlog at time $t$ in the reference FIFO system. 

Since the reference FIFO system has constant unit service rate with maximum packet length $\max_{n}\frac{L_n}{C_n}$, it has a service curve as $\hat{\beta}(t)=(t-\max_{n}\frac{L_n}{C_n})^{+}$. Since its input is constrained by a leaky-bucket arrival curve as $\hat{A}(s,t) = \sum_{n}\hat{A}_n(s,t) \le \sum_{n}\frac{r_n}{C_n} (t-s) + \sum_{n}\frac{\sigma_n}{C_n}$, the literature deterministic network calculus results give \cite{NetCal} $\hat{B}(t) \le \sum_{n}\frac{\sigma_n}{C_n}+ \sum_{n}\frac{r_n}{C_n}\max_{n}\frac{L_n}{C_n}$. The second part is finally obtained from $B(t) = C_{max} \sum_{n}\frac{B_n(t)}{C_{max}} \le C_{max} \sum_{n}\frac{B_n(t)}{C_{n}} = C_{max}\cdot \hat{B}(t)$.
\qed

\subsubsection{Proof of Theorem \ref{th-gr}}

For any packet $p^{f_n,i}$ of class $n$, we suppose its departure time $d^{f_n,i}$ is within the busy period that starts at $t^0$ and suppose $p^{f_n, i_0}$ is the first class $n$ packet served in this busy period. By these definitions, $a^{f,i_0} \ge t^0 .$ Note that this busy period always exists, since in the worst case, the period is only the service time period of $p^{f_n,i}$ and in this case, $t^0=a^{f_n,i}$ and $i_0=i$. 

Since the busy period starts at $t^0$, this implies that immediately before $t^0$, the system is idle. In other words, all packets, which arrived before $t^0$, have been served by $t^0$. Additionally since the node is work-conserving, FIFO and busy with serving between $t^0$ and $d^{g,j}$, there holds:
\begin{equation}\label{eq-3a}
d^{f_n,i} \le t^0 + \frac{\sum_{k=i_0}^{i} l^{f_n,k}}{C_n} + \sum_{m \neq n} \frac{A_{m}(t^0, a^{f_n,i})}{C_m},
\end{equation}
where, by definition, $A_{m}(t^0, a^{f_n,i})$ represents the total length (in bits) of packets from another flows $m \neq n$, which can maximally been served in $[t^0, a^{f_n,i}]$ before packet $p^{f_n, i}$. We then have:
\begin{eqnarray}
d^{f_n,i} &\le& t^0 + \frac{\sum_{k=i_0}^{i} l^{f_n,k}}{C_n} + \sum_{m \neq n} \frac{r_{m}(a^{f_n,i}-t^0) + \sigma_{m}}{C_m} \nonumber\\
&=& t_0 + \frac{\sum_{k=i_0}^{i} l^{f_n,k}}{C_n} + \sum_{m \neq n} \frac{\sigma_{m}}{C_m} + \sum_{m \neq n} \frac{r_{m}}{C_m} (a^{f_n,i}-t^0) \nonumber\\
&\le& t_0 + \frac{\sum_{k=i_0}^{i} l^{f_n,k}}{C_n} + \sum_{m \neq n} \frac{\sigma_{m}}{C_m} + \sum_{m \neq n} \frac{r_{m}}{C_m} (d^{f_n,i}-t^0) \nonumber\\
&\le& t_0 + \frac{\sum_{k=i_0}^{i} l^{f_n,k}}{(1-\bar{\rho}) C_n} + \sum_{m \neq n} \frac{\sigma_{m}}{(1-\bar{\rho})C_m} \nonumber
\end{eqnarray}
where $\bar{\rho} = \sum_{m \neq n} \frac{r_{m}}{C_m}$ and  the last step follows from some manipulation together with the left hand side.

Recall function (\ref{eq-grc}), it is easily verified that: 
\begin{eqnarray}
V^{f_n,i}((1-\bar{\rho}) C_n) &\ge& a^{f_n,i_0} + \frac{\sum_{k=i_0}^{i} l^{f,k}}{(1-\bar{\rho})C_n} \ge t_0 + \frac{\sum_{k=i_0}^{i} l^{f,k}}{(1-\bar{\rho})C_n} \nonumber 
\end{eqnarray}
Hence
$$
d^{f_n,i} \le V^{f,i}((1-\bar{\rho}) C_n) + \sum_{m \neq n} \frac{\sigma_{m}}{(1-\bar{\rho})C_m}.
$$
This ends the proof.
\qed

\subsubsection{Proof of Theorem \ref{th-losc}}
With Th.~\ref{th-gr}, Th.~\ref{th-losc} follows immediately from the relationship between GR and service curve \cite{Jiang03, NetCal}.

\subsection{Comparison}\label{sec-4b}
Comparing with Cor.~\ref{co-db} -- Cor.~\ref{co-losc}, Th.~\ref{th-db} -- Th.~\ref{th-losc} provide performance bounds better in two aspects. One is that the conditions for the latter are much relaxed. The new condition for the delay and backlog bounds, i.e. $\rho <1$, is indeed consistent with the multiclass FIFO stability condition found from the classic queueing theory (e.g. \cite{ChenZ97}). 

Another is that, the obtained bounds are also generally better. Specifically, the delay bound in Th. \ref{th-db} is tight in the sense that it may be reached by a packet as shown in the proof. In other words, there is no delay bound better than Th. \ref{th-db}, which may be found. In addition, it can be verified that the service curve obtained in Th.~\ref{th-losc} is always tighter than that in Cor.~\ref{co-losc}. Also, the same conclusion can be made for the guaranteed rate characterization. 

For backlog bound, the comparison is not straightforward. When $C_1 = \cdots = C_N$, it is easily verified that the bound in Th. \ref{th-bb} may only be slightly better due to the difference between $\max_n\frac{L_n}{C_n}$ and $\frac{L}{C_{min}}$. In general, when $C_n$ differs little, the bound in Th. \ref{th-bb} is better. However, when $C_n$ differs significantly, the bound in Cor.~\ref{co-bb} may become better but under the condition that $\sum_n r_n \le C_{min}$. If this condition is not met, then the bound in Th. \ref{th-bb} should be used.

\section{Concluding Remarks}\label{sec-5} 
In this paper, we have presented performance bounds for a multiclass FIFO system that has deterministically constrained traffic inputs. The analysis and results can be extended to obtain bounds for other performance metrics of interest, such as burstiness increase. Also, they can be extended to perform network case analysis by exploiting the GR and the service curve results. Moreover, the extension may be conducted for stochastic cases, which is our ongoing work. 

A final remark is that, many results for multiclass FIFO can be found under the classic queueing theory (see e.g. \cite{ChenZ97} and references therein). However, the arrival and service processes assumed in those results do not match with the traffic constraints considered in this paper. In the context of input-queueing in packet switching systems, multiclass FIFO has also been studied (e.g. \cite{Karol87}). However, in those studies, the focus has been on the throughput of the switch, \nop{ which has multiple input ports and multiple output ports, and on each input port, multclass FIFO may be employed, }assuming saturated traffic on each input port. 



\bibliographystyle{abbrv}
\bibliography{nc-qt}

\end{document}